\documentclass[fleqn]{llncs}
\usepackage{peis}

\title{Instruction Sequences\\ for the Production of Processes%
       \thanks{This research has been carried out as part of the project
               \emph{Thread Algebra for Strategic Interleaving}, which
               is funded by the Netherlands Organisation for Scientific
               Research (NWO).}}
\author{J.A. Bergstra \and C.A. Middelburg}
\institute{Programming Research Group, University of Amsterdam, \\
           Kruislaan~403, 1098~SJ~Amsterdam, the Netherlands \\
           \email{J.A.Bergstra@uva.nl,C.A.Middelburg@uva.nl}}

\begin{document}

\maketitle

\begin{abstract}
Single-pass instruction sequences under execution are considered to
produce behaviours to be controlled by some execution environment.
Threads as considered in thread algebra model such behaviours: upon each
action performed by a thread, a reply from its execution environment
determines how the thread proceeds.
Threads in turn can be looked upon as producing processes as considered
in process algebra.
We show that, by apposite choice of basic instructions, all processes
that can only be in a finite number of states can be produced by
single-pass instruction sequences.
\begin{classcode}
D.1.4, F.1.1, F.1.2, F.3.2.
\end{classcode}
\end{abstract}

\section{Introduction}
\label{sect-intro}

With the current paper, we carry on the line of research
with which a start was made in~\cite{BL02a}.
The working hypothesis of this line of research is that single-pass
instruction sequence is a central notion of computer science which
merits investigation for its own sake.
We take program algebra~\cite{BL02a} for the basis of our investigation.
Program algebra is a setting suited for investigating single-pass
instruction sequences.
It does not provide a notation for programs that is intended for actual
programming.

The starting-point of program algebra is the perception of a program as
a single-pass instruction sequence, i.e.\ a finite or infinite sequence
of instructions of which each instruction is executed at most once and
can be dropped after it has been executed or jumped over.
This perception is simple, appealing, and links up with practice.
Single-pass instruction sequences under execution are considered to
produce behaviours to be controlled by some execution environment.
Threads as considered in basic thread algebra~\cite{BL02a} model such
behaviours: upon each action performed by a thread, a reply from the
execution environment determines how the thread proceeds.%
\footnote
{In~\cite{BL02a}, basic thread algebra is introduced under the name
 basic polarized process algebra.
}
Threads in turn can be looked upon as producing processes as considered
in process algebras such as \ACP~\cite{BW90,Fok00} and CCS~\cite{Mil89}.
This means that single-pass instruction sequences under execution can be
considered to produce such processes.

Process algebra is considered relevant to computer science, as is
witnesses by the extent of the work on process algebra in theoretical
computer science.
This means that there must be programmed systems whose behaviours are
taken for processes as considered in process algebra.
This has motivated us to investigate the connections between programs
and the processes that they produce.
In this paper, we investigate those connections starting from the
perception of a program as a single-pass instruction sequence.

Regular threads are threads that can only be in a finite number of
states.
The behaviours of all single-pass instruction sequences considered in
program algebra are regular threads and all regular threads can be
produced by such single-pass instruction sequences.
Regular processes are processes that can only be in a finite number of
states.
We show in this paper that, by apposite choice of basic instructions,
all regular processes can be produced by such single-pass instruction
sequences as well.

To obtain this result naturally, we use single-pass instruction
sequences with multiple-reply test instructions, which are more general
than the test instructions considered in program algebra, and threads
with postconditional switching, which is more general than the
behavioural counterpart of test instructions considered in basic thread
algebra.
We show that the result can also be obtained without introducing
multiple-reply test instructions and postconditional switching if we
assume that the cluster fair abstraction rule (see e.g.~\cite{Fok00})
is valid.

\sloppy
Single-pass instruction sequences under execution, and more generally
threads, may make use of services such as counters, stacks and Turing
tapes.
The use operators introduced in~\cite{BM04c} are concerned with the
effect of services on threads.
An interesting aspect of making use of services is that it may turn a
regular thread into a non-regular thread.
Because non-regular threads produce non-regular processes, this means
that single-pass instruction sequences under execution that make use of
services may produce non-regular processes.
On that account, we add the use operators to basic thread algebra with
postconditional switching and make precise what processes are produced
by threads that make use of services.

Programs written in an assembly language are finite instruction
sequences for which single-pass execution is usually not possible.
However, the instruction set of such a program notation may be such that
all regular processes can as well be produced by programs written in
the program notation.
To illustrate this, we show that all regular processes can be produced
by programs written in a program notation which is close to existing
assembly languages.

This paper is organized as follows.
First, we review program algebra and extend it with multiple-reply test
instructions (Section~\ref{sect-PGA}).
Next, we review basic thread algebra, extend it with postconditional
switching (Section~\ref{sect-BTA}), and use the result to make
mathematically precise what threads are produced by the single-pass
instruction sequences considered in program algebra with multiple-reply
test instructions (Section~\ref{sect-thread-extr}).
Then, we review process algebra (Section~\ref{sect-ACP}) and use it to
make mathematically precise what processes are produced by the threads
considered in basic thread algebra with postconditional switching
(Section~\ref{sect-process-extr}).
After that, we show that all regular processes can be produced by the
single-pass instruction sequences considered in program algebra with
multiple-reply test instructions (Section~\ref{sect-expressiveness}).
Following this, we extend basic thread algebra  with postconditional
switching further to threads that make use of services and make precise
what processes are produced by such threads
(Section~\ref{sect-services}).
After that, we show that all regular processes can also be produced by
programs written in a program notation which is close to existing
assembly languages (Section~\ref{sect-PGLDmr-BR}).
Finally, we make some concluding remarks (Section~\ref{sect-concl}).

\section{Program Algebra with Multiple-Reply Test~Instructions}
\label{sect-PGA}

In this section, we first review \PGA\ (ProGram Algebra) and then extend
it with multiple-reply test instructions.
All regular processes can be produced by single-pass instruction
sequences as considered in \PGA\ extended with multiple-reply test
instructions provided use is made of basic instructions of a particular
kind.
Those basic instructions, which are called process construction
instructions, are also introduced.

\subsection{Program Algebra}
\label{subsect-PGA}

The perception of a program as a single-pass instruction sequence is the
starting-point of \PGA.

In \PGA, it is assumed that a fixed but arbitrary set $\BInstr$ of
\emph{basic instructions} has been given.
\PGA\ has the following \emph{primitive instructions}:
\begin{iteml}
\item
for each $a \in \BInstr$, a \emph{plain basic instruction} $a$;
\item
for each $a \in \BInstr$, a \emph{positive test instruction} $\ptst{a}$;
\item
for each $a \in \BInstr$, a \emph{negative test instruction} $\ntst{a}$;
\item
for each $l \in \Nat$, a \emph{forward jump instruction} $\fjmp{l}$;
\item
a \emph{termination instruction} $\halt$.
\end{iteml}
We write $\PInstr$ for the set of all primitive instructions of \PGA.

The intuition is that the execution of a basic instruction $a$ produces
either $\True$ or $\False$ at its completion.
In the case of a positive test instruction $\ptst{a}$, $a$ is executed
and execution proceeds with the next primitive instruction if $\True$ is
produced.
Otherwise, the next primitive instruction is skipped and execution
proceeds with the primitive instruction following the skipped one.
If there is no next instruction to be executed, deadlock occurs.
In the case of a negative test instruction $\ntst{a}$, the role of the
value produced is reversed.
In the case of a plain basic instruction $a$, execution always proceeds
as if $\True$ is produced.
The effect of a forward jump instruction $\fjmp{l}$ is that execution
proceeds with the $l$-th next instruction.
If $l$ equals $0$ or the $l$-th next instruction does not exist,
deadlock occurs.
The effect of the termination instruction $\halt$ is that execution
terminates.

\PGA\ has the following constants and operators:
\begin{iteml}
\item
for each $u \in \PInstr$, an \emph{instruction} constant $u$\,;
\item
the binary \emph{concatenation} operator $\conc$\,;
\item
the unary \emph{repetition} operator $\rep$\,.
\end{iteml}
We assume that there are infinitely many variables, including $X,Y,Z$.
Terms are built as usual.
We use infix notation for the concatenation operator and postfix
notation for the repetition operator.

A closed \PGA\ term is considered to denote a non-empty, finite or
periodic infinite sequence of primitive instructions.%
\footnote
{A periodic infinite sequence is an infinite sequence with only finitely
 many subsequences.}
Closed \PGA\ terms are considered equal if they denote the same
instruction sequence.
The axioms for instruction sequence equivalence are given in
Table~\ref{axioms-PGA}.%
\begin{table}[!t]
\caption{Axioms of \PGA}
\label{axioms-PGA}
\begin{eqntbl}
\begin{axcol}
(X \conc Y) \conc Z = X \conc (Y \conc Z)              & \axiom{PGA1} \\
(X^n)\rep = X\rep                                      & \axiom{PGA2} \\
X\rep \conc Y = X\rep                                  & \axiom{PGA3} \\
(X \conc Y)\rep = X \conc (Y \conc X)\rep              & \axiom{PGA4}
\end{axcol}
\end{eqntbl}
\end{table}
In this table, $n$ stands for an arbitrary natural number greater than
$0$.
For each \PGA\ term $P$, the term $P^n$ is defined by induction on $n$
as follows: $P^1 = P$ and $P^{n+1} = P \conc P^n$.
The \emph{unfolding} equation $X\rep = X \conc X\rep$ is derivable.
Each closed \PGA\ term is derivably equal to one of the form $P$ or
$P \conc Q\rep$, where $P$ and $Q$ are closed \PGA\ terms in which the
repetition operator does not occur.

Notice that PGA2 is actually an axiom schema.
Par abus de langage, axiom schemas will be called axioms throughout the
paper, with the exception of Section~\ref{sect-process-extr}.

\subsection{Multiple-Reply Test Instructions}
\label{subsect-mr-tests}

We introduce \PGAmr, an extension of \PGA\ with multiple-reply test
instructions.
These additional instructions are like the test instructions of \PGA,
but cover the case where a natural number greater than zero is produced
at the completion of the execution of a basic instruction.

In \PGAmr, like in \PGA, it is assumed that a fixed but arbitrary set
$\BInstr$ of basic instructions has been given.
\PGAmr\ has the primitive instructions of \PGA\ and in addition:
\begin{iteml}
\item
for each $n \in \Natpos$ and $a \in \BInstr$,
a \emph{positive multiple-reply test instruction}
$\pmrtst{n}{a}$\,;%
\footnote
{We write $\Natpos$ for the set $\set{n \in \Nat \where n > 0}$.}
\item
for each $n \in \Natpos$ and $a \in \BInstr$,
a \emph{negative multiple-reply test instruction}
$\nmrtst{n}{a}$\,.\nolinebreak\hsp{.325}
\end{iteml}
We write $\PInstrmr$ for the set of all primitive instructions of
\PGAmr.

The intuition is that the execution of a basic instruction $a$ produces
a natural number greater than zero at its completion.
In the case of a positive multiple-reply test instruction
$\pmrtst{n}{a}$, $a$ is executed and execution proceeds with the $i$-th
next primitive instruction if a natural number $i \leq n$ is produced.
If there is no next instruction to be executed or $i > n$, deadlock
occurs.
In the case of a negative multiple-reply test instruction
$\nmrtst{n}{a}$, execution proceeds with the $n{-}i{+}1$-th next
primitive instruction instead of the $i$-th one if a natural number
$i \leq n$ is produced.

For each $a \in \BInstr$, the instructions $\ptst{a}$ and $\ntst{a}$ are
considered essentially the same as the instructions $\pmrtst{2}{a}$ and
$\nmrtst{2}{a}$, respectively.
For that reason, the reply $\True$ is identified with the reply $1$ and
the reply $\False$ is identified with the reply~$2$.

\PGAmr\ has a constant $u$ for each $u \in \PInstrmr$.
The operators of \PGAmr\ are the same as the operators as \PGA.
Likewise, the axioms of \PGAmr\ are the same as the axioms as \PGA.

The intuition concerning multiple-reply test instructions given above
will be made fully precise in Section~\ref{sect-thread-extr}, using an
extension of basic thread algebra introduced in Section~\ref{sect-BTA}.

\subsection{Process Construction and Interaction with Services}
\label{subsect-pc-tsi}

Recall that, in \PGAmr, it is assumed that a fixed but arbitrary set
$\BInstr$ of basic instructions has been given.
In the sequel, we will make use a version of \PGAmr\ in which the
following additional assumptions relating to $\BInstr$ are made:
\begin{iteml}
\item
a fixed but arbitrary set $\Foci$ of \emph{foci} has been given;
\item
a fixed but arbitrary set $\Meth$ of \emph{methods} has been given;
\item
a fixed but arbitrary set $\AAct$ of \emph{atomic actions} has been
given;
\item
$\BInstr$ consists of:
\begin{iteml}
\item
for each $f \in \Foci$, $m \in \Meth$,
a \emph{program-service interaction instruction} $f.m$;
\item
for each $n \in \Natpos$, for each $e_1,\ldots,e_n \in \AAct$,
a \emph{process construction instruction} $\ac(e_1,\ldots,e_n)$.
\end{iteml}
\end{iteml}

Each focus plays the role of a name of some service provided by an
execution environment that can be requested to process a command.
Each method plays the role of a command proper.
Executing a basic instruction $f.m$ is taken as making a request to the
service named $f$ to process command $m$.

On execution of a basic instruction $\ac(e_1,\ldots,e_n)$, first a
non-deterministic choice between the atomic actions $e_1,\ldots,e_n$ is
made and then the chosen atomic action is performed.
The reply $1$ is produced if $e_1$ is performed, \ldots, the reply $n$
is produced if $e_n$ is performed.
Basic instructions of this kind are material to produce all regular
processes by single-pass instruction sequences.

We will write \PGAmrpc\ for the version of \PGAmr\ in which the
above-mentioned additional assumptions are made.

The intuition concerning program-service interaction instructions given
above will be made fully precise in Section~\ref{sect-services}, using
an extension of basic thread algebra.
The intuition concerning process construction instructions given above
will be made fully precise in Section~\ref{sect-process-extr}, using the
process algebra introduced in Section~\ref{sect-ACP}.
It will not be made fully precise using an extension of basic thread
algebra because it is considered a basic property of threads that they
are deterministic behaviours.

\section{Basic Thread Algebra with Postconditional Switching}
\label{sect-BTA}

In this section, we first review \BTA\ (Basic Thread Algebra) and then
extend it with postconditional switching.
All regular processes can be produced by threads as considered in
\BTA\ extended with postconditional switching provided use is made of
basic actions of a particular kind.
Those basic actions, which are the counterparts of the process
construction instructions from \PGAmrpc, are also introduced.

\subsection{Basic Thread Algebra}
\label{subsect-BTA}

\BTA\ is concerned with the behaviours that sequential programs exhibit
on execution.
These behaviours are called threads.

In \BTA, it is assumed that a fixed but arbitrary set $\BAct$ of
\emph{basic actions}, with $\Tau \notin \BAct$, has been given.
Besides, $\Tau$ is a special basic action.
We write $\BActTau$ for $\BAct \union \set{\Tau}$.
A thread performs basic actions in a sequential fashion.
Upon each basic action performed, a reply from the execution environment
of the thread determines how it proceeds.
The possible replies are $\True$ and $\False$.
Performing $\Tau$, which is considered performing an internal action,
always leads to the reply $\True$.

Although \BTA\ is one-sorted, we make this sort explicit.
The reason for this is that we will extend \BTA\ with an additional sort
in Section~\ref{sect-services}.

\BTA\ has one sort: the sort $\Thr$ of \emph{threads}.
To build terms of sort $\Thr$, it has the following constants and
operators:
\begin{iteml}
\item
the \emph{deadlock} constant $\const{\DeadEnd}{\Thr}$;
\item
the \emph{termination} constant $\const{\Stop}{\Thr}$;
\item
for each $a \in \BActTau$, the binary \emph{postconditional composition}
operator $\funct{\pccop{a}}{\Thr \x \Thr}{\Thr}$.
\end{iteml}
We assume that there are infinitely many variables of sort $\Thr$,
including $x,y,z$.
Terms of sort $\Thr$ are built as usual.
We use infix notation for the postconditional composition operator.
We introduce \emph{basic action prefixing} as an abbreviation:
$a \bapf p$ abbreviates $\pcc{p}{a}{p}$.

The thread denoted by a closed term of the form $\pcc{p}{a}{q}$ will
first perform $a$, and then proceed as the thread denoted by $p$
if the reply from the execution environment is $\True$ and proceed as
the thread denoted by $q$ if the reply from the execution environment is
$\False$.
The threads denoted by $\DeadEnd$ and $\Stop$ will become inactive and
terminate, respectively.

\BTA\ has only one axiom.
This axiom is given in Table~\ref{axioms-BTA}.%
\begin{table}[!tb]
\caption{Axiom of \BTA}
\label{axioms-BTA}
\begin{eqntbl}
\begin{axcol}
\pcc{x}{\Tau}{y} = \pcc{x}{\Tau}{x}                    & \axiom{T1}
\end{axcol}
\end{eqntbl}
\end{table}
Using the abbreviation introduced above, axiom T1 can be written as
follows: $\pcc{x}{\Tau}{y} = \Tau \bapf x$.

Notice that each closed \BTA\ term denotes a thread that will become
inactive or terminate after it has performed finitely many actions.
Infinite threads can be described by guarded recursion.

A \emph{guarded recursive specification} over \BTA\ is a set of
recursion equations $E = \set{X = t_X \where X \in V}$, where $V$ is a
set of variables of sort $\Thr$ and each $t_X$ is a \BTA\ term of the
form $\DeadEnd$, $\Stop$ or $\pcc{t}{a}{t'}$ with $t$ and $t'$ that
contain only variables from $V$.
We write $\vars(E)$ for the set of all variables that occur in $E$.
We are only interested in models of \BTA\ in which guarded recursive
specifications have unique solutions, such as the projective limit model
of \BTA\ presented in~\cite{BB03a}.

For each guarded recursive specification $E$ and each $X \in \vars(E)$,
we introduce a constant $\rec{X}{E}$ of sort $\Thr$ standing for the
unique solution of $E$ for $X$.
The axioms for these constants are given in Table~\ref{axioms-rec}.%
\begin{table}[!t]
\caption{Axioms for guarded recursion}
\label{axioms-rec}
\begin{eqntbl}
\begin{saxcol}
\rec{X}{E} = \rec{t_X}{E} & \mif X \!=\! t_X \in E       & \axiom{RDP}
\\
E \Implies X = \rec{X}{E} & \mif X \in \vars(E)          & \axiom{RSP}
\end{saxcol}
\end{eqntbl}
\end{table}
In this table, we write $\rec{t_X}{E}$ for $t_X$ with, for all
$Y \in \vars(E)$, all occurrences of $Y$ in $t_X$ replaced by
$\rec{Y}{E}$.\linebreak[2]
$X$, $t_X$ and $E$ stand for an arbitrary variable of sort $\Thr$, an
arbitrary \BTA\ term of sort $\Thr$ and an arbitrary guarded recursive
specification over \BTA, respectively.
Side conditions are added to restrict what $X$, $t_X$ and $E$ stand for.

Closed terms that denote the same infinite thread cannot always be
proved equal by means of the axioms given in Table~\ref{axioms-rec}.
We introduce \AIP\ (Approximation Induction Principle) to remedy this.
\AIP\ is based on the view that two threads are identical if their
approximations up to any finite depth are identical.
The approximation up to depth $n$ of a thread is obtained by cutting it
off after it has performed $n$ actions.
In \AIP, the approximation up to depth $n$ is phrased in terms of the
unary \emph{projection} operator $\funct{\projop{n}}{\Thr}{\Thr}$.
\AIP\ and the axioms for the projection operators are given in
Table~\ref{axioms-AIP}.%
\begin{table}[!t]
\caption{Approximation induction principle}
\label{axioms-AIP}
\begin{eqntbl}
\begin{axcol}
\AND{n \geq 0} \proj{n}{x} = \proj{n}{y} \Implies x = y & \axiom{AIP} \\
\proj{0}{x} = \DeadEnd                                  & \axiom{P0} \\
\proj{n+1}{\Stop} = \Stop                               & \axiom{P1} \\
\proj{n+1}{\DeadEnd} = \DeadEnd                         & \axiom{P2} \\
\proj{n+1}{\pcc{x}{a}{y}} =
                      \pcc{\proj{n}{x}}{a}{\proj{n}{y}} & \axiom{P3}
\end{axcol}
\end{eqntbl}
\end{table}
In this table, $a$ stands for an arbitrary bascic action from
$\BActTau$.

\subsection{Postconditional Switching}
\label{subsect-pcs}

We introduce \BTApcs, an extension of \BTA\ with postconditional
switching.
Postconditional switching is like postconditional composition, but
covers the case where the execution environment produces reply
values from the set $\Natpos$ instead of the set $\set{\True,\False}$.
Postconditional switching was first introduced in~\cite{BM08f}.

In \BTApcs, like in \BTA, it is assumed that a fixed but arbitrary set
$\BAct$ of basic actions, with $\Tau \notin \BAct$, has been given.
\BTApcs\ has the constants and operators of \BTA\ and in addition:
\begin{iteml}
\item
for each $a \in \BActTau$ and $k \in \Natpos$,
the $k$-ary \emph{postconditional switch} operator
$\funct{\pcsop{k}{a}}
  {\underbrace{\Thr \x \cdots \x \Thr}_{k \;\mathrm{times}}}
  {\Thr}$.
\end{iteml}

The thread denoted by a closed terms of the form
$\pcs{k}{a}{p_1,\ldots,p_k}$ will first perform $a$, and then proceed as
the thread denoted by $p_1$ if the processing of $a$ leads to the
reply~$1$, \ldots, proceed as the thread denoted by $p_k$ if the
processing of $a$ leads to the reply $k$.

For each $a \in \BActTau$, the operator $\pccop{a}$ is considered
essentially the same as the operator $\pcsop{2}{a}$.
For that reason, the reply $\True$ is identified with the reply $1$ and the
reply $\False$ is identified with the reply $2$.

Without additional assumptions about the set $\BAct$ of basic actions,
axioms S1 and T2 from Table~\ref{axioms-pcs} are the only axioms for
postconditional switching.
Axiom S1 expresses that the operators $\pccop{a}$ and $\pcsop{2}{a}$
are essentially the same.
Like axiom T1, axiom T2 reflects that performing $\Tau$ always leads to
the reply $1$.

Guarded recursion can be added to \BTApcs\ as it is added to \BTA\ in
Section~\ref{subsect-BTA}.

\subsection{Process Construction and Interaction with Services}
\label{subsect-pc-acts}

Recall that, in \BTApcs, it is assumed that a fixed but arbitrary set
$\BAct$ of basic actions has been given.
Like in the case of \PGAmr, we will make use in the sequel of a version
of \BTApcs\ in which the following additional assumptions relating to
$\BAct$ are made:
\begin{iteml}
\item
a fixed but arbitrary set $\Foci$ of \emph{foci} has been given;
\item
a fixed but arbitrary set $\Meth$ of \emph{methods} has been given;
\item
a fixed but arbitrary set $\AAct$ of \emph{atomic actions} has been
given;
\item
$\BAct$ consists of:
\begin{iteml}
\item
for each $f \in \Foci$ and $m \in \Meth$,
a \emph{thread-service interaction action} $f.m$;
\item
for each $n \in \Natpos$, for each $e_1,\ldots,e_n \in \AAct$,
a \emph{process construction action} $\ac(e_1,\ldots,e_n)$.
\end{iteml}
\end{iteml}

Like in the case of \PGAmr, performing a basic instruction $f.m$ is
taken as making a request to the service named $f$ to process command
$m$.

Like in the case of \PGAmr, on performing a basic action
$\ac(e_1,\ldots,e_n)$, first a non-deterministic choice between the
atomic actions $e_1,\ldots,e_n$ is made and then the chosen atomic
action is performed.
The reply $1$ is produced if $e_1$ is performed, \ldots, the reply $n$
is produced if $e_n$ is performed.

In Table~\ref{axioms-pcs}, axioms are given for the postconditional
switching operators which cover the case where the above-mentioned
additional assumptions about $\BAct$ are made.%
\begin{table}[!t]
\caption{Axioms for postconditional switching}
\label{axioms-pcs}
\begin{eqntbl}
\begin{saxcol}
\pcc{x}{a}{y} = \pcs{2}{a}{x,y}                        & & \axiom{S1} \\
\pcs{k}{\ac(e_1,\ldots,e_n)}{x_1,\ldots,x_k} =
\pcs{n}{\ac(e_1,\ldots,e_n)}{x_1,\ldots,x_n}
                                            & \mif n < k & \axiom{S2} \\
\pcs{k}{\ac(e_1,\ldots,e_n)}{x_1,\ldots,x_k} =
\pcs{n}{\ac(e_1,\ldots,e_n)}{x_1,\ldots,x_k,
\underbrace{\DeadEnd,\ldots,\DeadEnd}_{n-k \,\mathrm{times}}}
                                            & \mif n > k & \axiom{S3} \\
\pcs{k}{\Tau}{x_1,\ldots,x_k} =
\smash{\pcs{k}{\Tau}{\overbrace{x_1,\ldots,x_1}^{k \,\mathrm{times}}}}
                                                       & & \axiom{T2}
\end{saxcol}
\end{eqntbl}
\end{table}
In this table, $a$ stands for an arbitrary basic action from $\BActTau$
and $e_1,\ldots,e_n$ stand for arbitrary atomic actions from $\AAct$.

Axioms S2 and S3 stipulate that a thread denoted by a term of the
form $\pcs{k}{\ac(e_1,\ldots,e_n)}{p_1,\ldots,p_k}$ behaves as if it
concerns a $n$-ary postconditional switch if $n \neq k$.
The $n$-ary postconditional switch in question is obtained by removing
$p_{n+1},\ldots,p_k$ if $n < k$, and is obtained by adding $\DeadEnd$
sufficiently many times if $n > k$.

We will write \BTApcspc\ for the version of \BTApcs\ in which the
above-men\-tioned additional assumptions are made.

\section{Thread Extraction}
\label{sect-thread-extr}

In this short section, we use \BTApcs\ with guarded recursion to make
mathematically precise what threads are produced by the single-pass
instruction sequences denoted by closed \PGAmr\ terms.

The \emph{thread extraction} operation $\textr{\ph}$ determines, for
each closed \PGAmr\ term $P$, a closed term of \BTApcs\ with guarded
recursion that denotes the thread produced by the single-pass
instruction sequence denoted by $P$.
The thread extraction operation is defined by the equations given in
Table~\ref{axioms-thread-extr} (for $a \in \BInstr$, $n \in \Natpos$,
$l \in \Nat$, and $u \in \PInstrmr$)%
\begin{table}[!t]
\caption{Defining equations for thread extraction operation}
\label{axioms-thread-extr}
\begin{eqntbl}
\begin{eqncol}
\textr{a} = a \bapf \DeadEnd \\
\textr{a \conc X} = a \bapf \textr{X} \\
\textr{\ptst{a}} = a \bapf \DeadEnd \\
\textr{\ptst{a} \conc X} =
\pcc{\textr{X}}{a}{\textr{\fjmp{2} \conc X}} \\
\textr{\ntst{a}} = a \bapf \DeadEnd \\
\textr{\ntst{a} \conc X} =
\pcc{\textr{\fjmp{2} \conc X}}{a}{\textr{X}} \\
\textr{\pmrtst{n}{a}} = a \bapf \DeadEnd \\
\textr{\pmrtst{n}{a} \conc X} =
\pcs{n}{a}{\textr{\fjmp{1} \conc X},\ldots,\textr{\fjmp{n} \conc X}} \\
\textr{\nmrtst{n}{a}} = a \bapf \DeadEnd \\
\textr{\nmrtst{n}{a} \conc X} =
\pcs{n}{a}{\textr{\fjmp{n} \conc X},\ldots,\textr{\fjmp{1} \conc X}}
\end{eqncol}
\qquad
\begin{eqncol}
\textr{\fjmp{l}} = \DeadEnd \\
\textr{\fjmp{0} \conc X} = \DeadEnd \\
\textr{\fjmp{1} \conc X} = \textr{X} \\
\textr{\fjmp{l+2} \conc u} = \DeadEnd \\
\textr{\fjmp{l+2} \conc u \conc X} = \textr{\fjmp{l+1} \conc X} \\
\\
\textr{\halt} = \Stop \\
\textr{\halt \conc X} = \Stop
\end{eqncol}
\end{eqntbl}
\end{table}
and the rule that $\textr{\fjmp{l} \conc X} = \DeadEnd$ if $\fjmp{l}$ is
the beginning of an infinite jump chain.
This rule is formalized in e.g.~\cite{BM07g}.

The equations in Table~\ref{axioms-thread-extr} relating to the
primitive instructions of \PGA\ are the equations that have been used to
define the thread extraction operation for \PGA\ in most earlier work on
\PGA\ (see e.g.~\cite{BM07g,PZ06a}).
The additional equations relating to multiple-reply test instructions
are obvious generalizations of the equations relating to the test
instructions of \PGA.

Let $P$ be a closed \PGAmr\ term.
Then we say that $\textr{P}$ is the \emph{thread produced by} $P$.

\section{Process Algebra}
\label{sect-ACP}

In this section, we review \ACPt\ (Algebra of Communicating Processes
with abstraction).
This is the process algebra that will be used in
Section~\ref{sect-process-extr} to make precise what processes are
produced by the single-pass instruction sequences denoted by closed
\PGAmrpc\ terms.

In \ACPt, it is assumed that a fixed but arbitrary set $\Act$ of
\emph{atomic actions}, with $\tau,\dead \notin \Act$, and a fixed but
arbitrary commutative and associative function
$\funct{\commm}{\Act \union \set{\tau} \x \Act \union \set{\tau}}
               {\Act \union \set{\dead}}$,
with $\tau \commm e = \dead$ for all $e \in \Act \union \set{\tau}$,
have been given.
The function $\commm$ is regarded to give the result of synchronously
performing any two atomic actions for which this is possible, and to
give $\dead$ otherwise.
In \ACPt, $\tau$ is a special atomic action, called the silent step.
The act of performing the silent step is considered unobservable.
Because it would otherwise be observable, the silent step is considered
an atomic action that cannot be performed synchronously with other
atomic actions.
We write $\Actt$ for $\Act \union \set{\tau}$.

\ACPt\ has the following constants and operators:
\begin{itemize}
\item
for each $e \in \Act$, the \emph{atomic action} constant $e$\,;
\item
the \emph{silent step} constant $\tau$\,;
\item
the \emph{deadlock} constant $\dead$\,;
\item
the binary \emph{alternative composition} operator $\altc$\,;
\item
the binary \emph{sequential composition} operator $\seqc$\,;
\item
the binary \emph{parallel composition} operator $\parc$\,;
\item
the binary \emph{left merge} operator $\leftm$\,;
\item
the binary \emph{communication merge} operator $\commm$\,;
\item
for each $H \subseteq \Act$, the unary \emph{encapsulation} operator
$\encap{H}$\,;
\item
for each $I \subseteq \Act$, the unary \emph{abstraction} operator
$\abstr{I}$\,.
\end{itemize}
We assume that there are infinitely many variables.
Terms are built as usual.
We use infix notation for the binary operators.

Let $p$ and $q$ be closed \ACPt\ terms, $e \in \Act$, and
$H,I \subseteq \Act$.
Intuitively, the constants and operators to build \ACPt\ terms can be
explained as follows:
\begin{itemize}
\item
$e$ first performs atomic action $e$ and next terminates successfully;
\item
$\tau$ performs an unobservable atomic action and next terminates
successfully;
\item
$\dead$ can neither perform an atomic action nor terminate successfully;
\item
$p \altc q$ behaves either as $p$ or as $q$, but not both;
\item
$p \seqc q$ first behaves as $p$ and on successful termination of $p$
it next behaves as~$q$;
\item
$p \parc q$ behaves as the process that proceeds with $p$ and $q$ in
parallel;
\item
$p \leftm q$ behaves the same as $p \parc q$, except that it starts
with performing an atomic action of $p$;
\item
$p \commm q$ behaves the same as $p \parc q$, except that it starts with
performing an\linebreak[2] atomic action of $p$ and an atomic action of
$q$ synchronously;
\item
$\encap{H}(p)$ behaves the same as $p$, except that atomic actions from
$H$ are blocked;
\item
$\abstr{I}(p)$ behaves the same as $p$, except that atomic actions from
$I$ are turned into unobservable atomic actions.
\end{itemize}
The operators $\leftm$ and $\commm$ are of an auxiliary nature.
They are needed to axiomatize \ACPt.
The axioms of \ACPt\ are given in e.g.~\cite{Fok00}.

We write $\vAltc{i \in S} p_i$, where $S = \set{i_1,\ldots,i_n}$ and
$p_{i_1},\ldots,p_{i_n}$ are \ACPt\ terms,
for $p_{i_1} \altc \ldots \altc p_{i_n}$.
The convention is that $\vAltc{i \in S} p_i$ stands for $\dead$ if
$S = \emptyset$.

A \emph{recursive specification} over \ACPt\ is a set of recursion
equations $E = \set{X = t_X \where X \in V}$, where $V$ is a set of
variables and each $t_X$ is an \ACPt\ term containing only variables
from $V$.
Let $t$ be an \ACPt\ term without occurrences of abstraction operators
containing a variable $X$.
Then an occurrence of $X$ in $t$ is \emph{guarded} if $t$ has a subterm
of the form $e \seqc t'$ where $e \in \Act$ and $t'$ is a term
containing this occurrence of $X$.
Let $E$ be a recursive specification over \ACPt.
Then $E$ is a \emph{guarded recursive specification} if, in each
equation $X = t_X \in E$:
(i)~abstraction operators do not occur in $t_X$ and
(ii)~all occurrences of variables in $t_X$ are guarded or $t_X$ can be
rewritten to such a term using the axioms of \ACPt\ in either direction
and/or the equations in $E$ except the equation $X = t_X$ from left to
right.
We only consider models of \ACPt\ in which guarded recursive
specifications have unique solutions, such as the models of \ACPt\
presented in~\cite{BW90}.

For each guarded recursive specification $E$ and each variable $X$ that
occurs in $E$, we introduce a constant $\rec{X}{E}$ standing for the
unique solution of $E$ for $X$.
The axioms for these constants are given in~\cite{Fok00}.

\section{Process Extraction}
\label{sect-process-extr}

In this section, we use \ACPt\ with guarded recursion to make
mathematically precise what processes are produced by the single-pass
instruction sequences denoted by closed \PGAmrpc\ terms.

For that purpose, $\Act$ and $\commm$ are taken such that:
\begin{ldispl}
\begin{aeqns}
\multicolumn{3}{@{}l@{}}{\AAct \subseteq \Act\;,}
\\
\Act \diff \AAct  & = &
\set{\snd_f(d) \where f \in \Foci, d \in \Meth \union \Nat} \union
\set{\rcv_f(d) \where f \in \Foci, d \in \Meth \union \Nat}
\\ & {} \union {} &
\set{\snd_\serv(r) \where r \in \Nat} \union
\set{\rcv_\serv(m) \where m \in \Meth} \union
\set{\stp,\ol{\stp},\stp^*,\iact}
\end{aeqns}
\end{ldispl}%
and for all $e,e' \in \Act$, $f \in \Foci$, $d \in \Meth \union \Nat$,
$m \in \Meth$, and $r \in \Nat$:
\begin{ldispl}
\begin{aeqns}
\snd_f(d) \commm \rcv_f(d) = \iact \;,
\\
\snd_f(d) \commm e = \dead & & \mif e \neq \rcv_f(d)\;,
\\
e \commm \rcv_f(d) = \dead & & \mif e \neq \snd_f(d)\;,
\eqnsep
\snd_\serv(r) \commm e = \dead\;,
\\
e \commm \rcv_\serv(m) = \dead\;,
\end{aeqns}
\qquad\;
\begin{aeqns}
\stp \commm \ol{\stp} = \stp^*\;,
\\
\stp \commm e = \dead      & & \mif e \neq \ol{\stp}\;,
\\
e \commm \ol{\stp} = \dead & & \mif e \neq \stp\;,
\eqnsep
\iact \commm e = \dead\;,
\\
e' \commm e = \dead & & \mif e' \in \AAct\;.
\end{aeqns}
\end{ldispl}

The \emph{process extraction} operation $\pextr{\ph}$ determines, for
each closed \BTApcspc\ term $p$, a closed term of \ACPt\ with guarded
recursion that denotes the process produced by the thread denoted by
$p$.
The process extraction operation $\pextr{\ph}$ is defined by
$\pextr{p} = \abstr{\set{\stp}}(\cpextr{p})$, where $\cpextr{\ph}$ is
defined by the equations given in Table~\ref{eqns-process-extr}
(for $f \in \Foci$, $m \in \Meth$, and $e_1,\ldots,e_n \in \AAct$).%
\begin{table}[!t]
\caption{Defining equations for process extraction operation}
\label{eqns-process-extr}
\begin{eqntbl}
\begin{seqncol}
\cpextr{X} = X
\\
\cpextr{\Stop} = \stp
\\
\cpextr{\DeadEnd} = \iact \seqc \dead
\\
\cpextr{\pcc{t_1}{\Tau}{t_2}} = \iact \seqc \iact \seqc \cpextr{t_1}
\\
\cpextr{\pcc{t_1}{f.m}{t_2}} =
\snd_f(m) \seqc
(\rcv_f(1) \seqc \cpextr{t_1} \altc \rcv_f(2) \seqc \cpextr{t_2})
\\
\cpextr{\pcc{t_1}{\ac(e_1,\ldots,e_n)}{t_2}} =
e_1 \seqc \cpextr{t_1} \altc
e_2 \seqc \cpextr{t_2} \altc \ldots \altc e_n \seqc \cpextr{t_2}
\\
\cpextr{\pcs{k}{\Tau}{t_1,\ldots,t_k}} =
\iact \seqc \iact \seqc \cpextr{t_1}
\\
\cpextr{\pcs{k}{f.m}{t_1,\ldots,t_k}} =
\snd_f(m) \seqc
(\rcv_f(1) \seqc \cpextr{t_1} \altc \ldots \altc
 \rcv_f(k) \seqc \cpextr{t_k})
\\
\cpextr{\pcs{k}{\ac(e_1,\ldots,e_n)}{t_1,\ldots,t_k}} =
e_1 \seqc \cpextr{t_1} \altc \ldots \altc e_n \seqc \cpextr{t_n}
 & \mif n \leq k
\\
\cpextr{\pcs{k}{\ac(e_1,\ldots,e_n)}{t_1,\ldots,t_k}} =
{} \\ \qquad
e_1 \seqc \cpextr{t_1} \altc \ldots \altc e_k \seqc \cpextr{t_k} \altc
e_{k+1} \seqc \iact \seqc \dead \altc \ldots \altc
e_n \seqc \iact \seqc \dead
 & \mif n > k
\\
\cpextr{\rec{X}{E}} =
\rec{X}{\set{X' = \cpextr{t_{X'}} \where X' = t_{X'} \,\in\, E}}
\end{seqncol}
\end{eqntbl}
\end{table}

Two atomic actions are involved in performing a basic action of the form
$f.m$: one for sending a request to process command $m$ to the service
named $f$ and another for receiving a reply from that service upon
completion of the processing.
Performing a basic action of the form $\ac(e_1,\ldots,e_n)$ always gives
rise to a non-deterministic choice between $n$ alternatives, where $e_i$
is the first atomic action of the $i$-th alternatives.

For each closed \BTApcspc\ term $p$, $\cpextr{p}$ denotes a process that
will perform a special termination action just before successful
termination.
Abstracting from this termination action yields the process denoted by
$\pextr{p}$.
In Section~\ref{sect-services}, \BTApcspc\ is extended with use
operators, which are concerned with threads making use of services.
The process extraction operation $\pextr{\ph}$ for \BTApcspc\ is defined
here in terms of $\cpextr{\ph}$ to allow for the process extraction
operation for the extension of \BTApcspc\ with use operators to be
defined easily.

Some actions introduced above are not used in the definition of the
process extraction operation for \BTApcspc.
Those actions are used in the definition of the process extraction
operation for the extension of \BTApcspc\ with use operators.

Let $p$ be a closed \BTApcspc\ term and $P$ be a closed \PGAmrpc\ term.
Then we say that $\pextr{p}$ is the \emph{process produced by} $p$ and
$\pextr{\textr{P}}$ is the \emph{process produced by} $P$.

\sloppy
The process extraction operation preserves the axioms of \BTApcspc\ with
guarded recursion.
Roughly speaking, this means that the translations of these axioms are
derivable from the axioms of \ACPt\ with guarded recursion.
Before we make this fully precise, we have a closer look at the axioms
of \BTApcspc\ with guarded recursion.

A proper axiom is an equation or a conditional equation.
In Tables~\ref{axioms-rec} and~\ref{axioms-pcs}, we do not find proper
axioms.
Instead of proper axioms, we find axiom schemas without side conditions
and axiom schemas with syntactic side conditions.
The axioms of \BTApcspc\ with guarded recursion are obtained by
replacing each axiom schema by all its instances.

We define a function $\transl{\ph}$ from the set of all equations and
conditional equations of \BTApcspc\ with guarded recursion to the set of
all equations of \ACPt\ with guarded recursion as follows:
\begin{ldispl}
\transl{t_1 = t_2} \;\;=\;\; \pextr{t_1} = \pextr{t_2}\;,
\\
\transl{E \Implies t_1 = t_2} \;\;=\;\;
\set{\pextr{t'_1} = \pextr{t'_2} \where t'_1 = t'_2 \,\in\, E} \Implies
\pextr{t_1} = \pextr{t_2}\;.
\end{ldispl}
\begin{proposition}
\label{prop-preservation-axioms}
Let $\phi$ be an axiom of \BTApcspc\ with guarded recursion.
Then $\transl{\phi}$ is derivable from the axioms of \ACPt\ with guarded
recursion.
\end{proposition}
\begin{proof}
The proof is trivial.
\qed
\end{proof}
Proposition~\ref{prop-preservation-axioms} would go through if no
abstraction of the above-mentioned special termination action was made.
However, the expressiveness results for \PGAmrpc\ relating to processes
that are presented in Section~\ref{sect-expressiveness} would not go
through.
Notice further that \ACPt\ without the silent step constant and the
abstraction operator, better known as \ACP, would suffice if no
abstraction of the special termination action was made.

\section{Expressiveness of \PGAmrpc}
\label{sect-expressiveness}

In this section, we show that all regular processes can be produced by
the single-pass instruction sequences considered in program algebra with
multiple-reply test instructions.

We begin by making precise what it means that a thread can only be in a
finite number of states.
We assume that a fixed but arbitrary model $\fM$ of \BTApcs\ extended
with guarded recursion has been given, we use the term thread only for
the elements from the domain of $\fM$, and we denote the interpretations
of constants and operators in $\fM$ by the constants and operators
themselves.

Let $p$ be a thread.
Then the set of \emph{states} or \emph{residual threads} of $p$,
written $\Res(p)$, is inductively defined as follows:
\begin{itemize}
\item
$p \in \Res(p)$;
\item
if $\pcc{q}{a}{r} \in \Res(p)$, then $p,q \in \Res(p)$;
\item
if $\pcs{k}{a}{p_1,\ldots,p_k} \in \Res(p)$, then
$p_1,\ldots,p_k \in \Res(p)$.
\end{itemize}

Let $p$ be a thread and let $\BAct' \subseteq \BActTau$.
Then $p$ is \emph{regular over} $\BAct'$ if the following conditions are
satisfied:
\begin{iteml}
\item
$\Res(p)$ is finite;
\item
for all $q,r \in \Res(p)$ and $a \in \BActTau$,
$\pcc{q}{a}{r} \in \Res(p)$ implies $a \in \BAct'$;
\item
for all $p_1,\ldots,p_k \in \Res(p)$ and $a \in \BActTau$,
$\pcs{k}{a}{p_1,\ldots,p_k} \in \Res(p)$ implies $a \in \BAct'$.
\end{iteml}
We say that $p$ is \emph{regular} if $p$ is regular over $\BActTau$.
\pagebreak[2]

We will make use of the fact that being a regular thread coincides with
being the solution of a finite guarded recursive specification in which
the right-hand sides of the recursion equations are of a restricted
form.

A \emph{linear recursive specification} over \BTApcs\ is a guarded
recursive specification $E = \set{X = t_X \where X \in V}$ over
\BTApcs, where each $t_X$ is a term of the form $\DeadEnd$, $\Stop$,
$\pcc{Y}{a}{Z}$ with $Y,Z \in V$ or $\pcs{k}{a}{X_1,\ldots,X_k}$ with
$X_1,\ldots,X_k \in V$.
\begin{proposition}
\label{prop-lin-rec-thread}
Let $p$ be a thread and let $\BAct' \subseteq \BActTau$.
Then $p$ is regular over $\BAct'$ iff there exists a finite linear
recursive specification $E$ over \BTApcs\ in which only basic actions
from $\BAct'$ occur such that $p$ is the solution of $E$ for some
$X \in \vars(E)$.
\end{proposition}
\begin{proof}
This proposition generalizes Theorem~1 from~\cite{PZ06a} from \BTA\ to
\BTApcs\ and from the projective limit model of \BTA\ to an arbitrary
model of \BTApcs.
However, the proof of that theorem is applicable to any model of \BTA\
and the adaptations needed to take postconditional switching
operators and their interpretations into account are trivial.
\qed
\end{proof}

All regular threads over $\BAct$ can be produced by the single-pass
instruction sequences considered in program algebra with multiple-reply
test instructions.
\begin{proposition}
\label{prop-expressiveness}
For each thread $p$ that is regular over $\BAct$, there exists a closed
\PGAmr\ term $P$ such that $p$ is the thread denoted by $\textr{P}$.
\end{proposition}
\begin{proof}
This proposition generalizes one direction of Proposition~2
from~\cite{PZ06a} from \PGA\ to \PGAmr\ and from the projective limit
model of \BTA\ to an arbitrary model of \BTApcs.
However, the proof of that proposition is applicable to any model of
\BTA\ and the adaptations needed to take multiple-reply test
instructions and the interpretations of postconditional switching
operators into account are trivial.
\qed
\end{proof}

We proceed by making precise what it means that a process can only be in
a finite number of states.
We assume that a fixed but arbitrary model $\fM'$ of \ACPt\ with guarded
recursion has been given, we use the term process only for the elements
from the domain of $\fM'$, and we denote the interpretations of constants
and operators in $\fM'$ by the constants and operators themselves.

Let $p$ be a process.
Then the set of \emph{states} or \emph{subprocesses} of $p$,
written $\Sub(p)$, is inductively defined as follows:
\begin{itemize}
\item
$p \in \Sub(p)$;
\item
if $e \seqc q \in \Sub(p)$, then $q \in \Sub(p)$;
\item
if $e \seqc q \altc r \in \Sub(p)$, then $q \in \Sub(p)$.
\end{itemize}

Let $p$ be a process and let $\Act' \subseteq \Actt$.
Then $p$ is \emph{regular over} $\Act'$ if the following conditions are
satisfied:
\begin{iteml}
\item
$\Sub(p)$ is finite;
\item
for all $q \in \Sub(p)$ and $e \in \Actt$,
$e \seqc q \in \Sub(p)$ implies $e \in \Act'$;
\item
for all $q,r \in \Sub(p)$ and $e \in \Actt$,
$e \seqc q \altc r \in \Sub(p)$ implies $e \in \Act'$.
\end{iteml}
We say that $p$ is \emph{regular} if $p$ is regular over $\Actt$.

We will make use of the fact that being a regular process over $\Act$
coincides with being the solution of a finite guarded recursive
specification in which the right-hand sides of the recursion equations
are linear terms.
\emph{Linearity} of terms is inductively defined as follows:
\begin{iteml}
\item
$\dead$ is linear;
\item
if $e \in \Actt$, then $e$ is linear;
\item
if $e \in \Actt$ and $X$ is a variable, then $e \seqc X$ is linear;
\item
if $t$ and $t'$ are linear, then $t \altc t'$  is linear.
\end{iteml}
A \emph{linear recursive specification} over \ACPt\ is a guarded
recursive specification $E = \set{X = t_X \where X \in V}$ over \ACPt,
where each $t_X$ is linear.
\begin{proposition}
\label{prop-lin-rec-process}
Let $p$ be a process and let $\Act' \subseteq \Act$.
Then $p$ is regular over $\Act'$ iff there exists a finite linear
recursive specification $E$ over \ACPt\ in which only atomic actions
from $\Act'$ occur such that $p$ is the solution of $E$ for some
$X \in \vars(E)$.
\end{proposition}
\begin{proof}
The proof follows the same line as the proof of
Proposition~\ref{prop-lin-rec-thread}.
\qed
\end{proof}
\begin{unnumremark}
Proposition~\ref{prop-lin-rec-process} is concerned with processes that
are regular over $\Act$.
We can also prove that being a regular process over $\Actt$ coincides
with being the solution of a finite linear recursive specification over
\ACPt if we assume that the cluster fair abstraction rule~\cite{Fok00}
holds in the model $\fM'$.
However, we do not need this more general result.
\end{unnumremark}

All regular processes over $\AAct$ can be produced by the single-pass
instruction sequences considered in program algebra with multiple-reply
test instructions.
\begin{theorem}
\label{theorem-expressiveness}
For each process $p$ that is regular over $\AAct$, there exists a closed
\PGAmrpc\ term $P$ such that $p$ is the process denoted by
$\pextr{\textr{P}}$.
\end{theorem}
\begin{proof}
By Propositions~\ref{prop-lin-rec-thread}, \ref{prop-expressiveness}
and~\ref{prop-lin-rec-process}, it is sufficient to show that, for each
finite linear recursive specification $E$ over \ACPt\ in which only
atomic actions from $\AAct$ occur, there exists a finite linear
recursive specification $E'$ over \BTApcspc\ such that $\rec{X}{E} =
\pextr{\rec{X}{E'}}$ for all $X \in \vars(E)$.

Take the finite linear recursive specification $E$ over \ACPt\ that
consists of the recursion equations
\begin{ldispl}
X_i = e_{i 1} \seqc X_{i 1} \altc \ldots \altc e_{i k_i} \seqc X_{i k_i}
       \altc
      e'_{i 1} \altc \ldots \altc e'_{i l_i}\;,
\end{ldispl}
where $e_{i 1},\ldots,e_{i k_i},e'_{i 1},\ldots,e'_{i l_i} \in \AAct$,
for $i \in \set{1,\ldots n}$.
Then construct the finite linear recursive specification $E'$ over
\BTApcspc\ that consists of the recursion equations
\begin{ldispl}
X_i =
\pcs{k_i+l_i}
    {\ac(e_{i 1},\ldots,e_{i k_i},e'_{i 1},\ldots,e'_{i l_i})}
    {X_{i 1},\ldots,X_{i k_i},
     \underbrace{\Stop,\ldots,\Stop}_{l_i \,\mathrm{times}}}
\end{ldispl}
for $i \in \set{1,\ldots n}$.
It follows immediately from the definition of the process extraction
operation that $\rec{X}{E} = \pextr{\rec{X}{E'}}$ for all
$X \in \vars(E)$.
\qed
\end{proof}
Multiple-reply test instructions and postconditional switching have been
introduced because process construction instructions of the form
$\ac(e_1,\ldots,e_n)$ with $n > 2$ look to be necessary to obtain this
result.
However, a similar result can also be obtained for closed \PGAmrpc\
terms in which only basic instructions of the form $\ac(e_1,e_2)$ occur
if we assume that the cluster fair abstraction rule~\cite{Fok00} holds
in the model $\fM'$.
\begin{theorem}
\label{theorem-completeness}
Assume that CFAR (Cluster Fair Abstraction Rule) holds in $\fM'$.
Let $\tact \in \AAct$.
Then, for each process $p$ that is regular over
$\AAct \diff \set{\tact}$, there exists a closed \PGAmrpc\ term $P$ in
which only basic instructions of the form $\ac(e,\tact)$ occur such that
$\tau \seqc p$ is the process denoted by
$\tau \seqc \abstr{\set{\tact}}(\pextr{\textr{P}})$.
\end{theorem}
\begin{proof}
By Propositions~\ref{prop-lin-rec-thread}, \ref{prop-expressiveness}
and~\ref{prop-lin-rec-process} and the definition of the thread
extraction operation, it is sufficient to show that, for each finite
linear recursive specification $E$ over \ACPt\ in which only atomic
actions from $\AAct \diff \set{\tact}$ occur, there exists a finite
linear recursive specification $E'$ over \BTApcspc\ in which only basic
actions of the form $\ac(e,\tact)$ occur such that
$\tau \seqc \rec{X}{E} =
 \tau \seqc \abstr{\set{\tact}}(\pextr{\rec{X}{E'}})$
for all $X \in \vars(E)$.

Take the finite linear recursive specification $E$ over \ACPt\ that
consists of the recursion equations
\begin{ldispl}
X_i = e_{i 1} \seqc X_{i 1} \altc \ldots \altc e_{i k_i} \seqc X_{i k_i}
       \altc
      e'_{i 1} \altc \ldots \altc e'_{i l_i}\;,
\end{ldispl}
where
$e_{i 1},\ldots,e_{i k_i},e'_{i 1},\ldots,e'_{i l_i} \in
 \AAct \diff \set{\tact}$,
for $i \in \set{1,\ldots n}$.
Then construct the finite linear recursive specification $E'$ over
\BTApcspc\ that consists of the recursion equations
\begin{ldispl}
X_i =
\pcc{X_{i1}}{\ac(e_{i1},\tact)}
    {(\ldots(\pcc{X_{ik_i}}{\ac(e_{ik_i},\tact)}
                 {\\ \hsp{3.1} (\pcc{\Stop}{\ac(e'_{i1},\tact)}
                       {(\ldots(\pcc{\Stop}{\ac(e'_{il_i},\tact)}
                                    {X_i})\ldots)})})\ldots)}
\end{ldispl}
for $i \in \set{1,\ldots n}$;
and the finite linear recursive specification $E''$ over \ACPt\ that
consists of the recursion equations
\begin{ldispl}
\begin{aeqns}
X_i       & = & e_{i 1}   \seqc X_{i 1}   \altc \tact \seqc Y_{i 2}\;,
\\
Y_{i 2}   & = & e_{i 2}   \seqc X_{i 2}   \altc \tact \seqc Y_{i 3}\;,
\\ & \vdots & \\
Y_{i k_i} & = & e_{i k_i} \seqc X_{i k_i} \altc \tact \seqc Z_{i 1}\;,
\end{aeqns}
\qquad \qquad
\begin{aeqns}
Z_{i 1}   & = & e'_{i 1}   \altc \tact \seqc Z_{i 2}\;,
\\
Z_{i 2}   & = & e'_{i 2}   \altc \tact \seqc Z_{i 3}\;,
\\ & \vdots & \\
Z_{i l_i} & = & e'_{i l_i} \altc \tact \seqc X_i\;,
\end{aeqns}
\end{ldispl}
where $Y_{i 2},\ldots,Y_{i k_i},Z_{i 1},\ldots,Z_{i l_i}$ are fresh
variables, for $i \in \set{1,\ldots n}$.
It follows immediately from the definition of the process extraction
operation that $\pextr{\rec{X}{E'}} = \rec{X}{E''}$ for all
$X \in \vars(E)$.
Moreover, it follows from CFAR that
$\tau \seqc \rec{X}{E} = \tau \seqc \abstr{\set{\tact}}(\rec{X}{E''})$
for all $X \in \vars(E)$.
Hence,
$\tau \seqc \rec{X}{E} =
 \tau \seqc \abstr{\set{\tact}}(\pextr{\rec{X}{E'}})$
for all $X \in \vars(E)$.
\qed
\end{proof}

\section{Services}
\label{sect-services}

An instruction sequence under execution may make use of services.
That is, certain instructions may be executed for the purpose of having
the behaviour produced by the instruction sequence affected by a
service that takes those instructions as commands to be processed.
Likewise, a thread may perform certain actions for the purpose of having
itself affected by a service that takes those actions as commands to be
processed.
The processing of an action may involve a change of state of the service
and at completion of the processing of the action the service returns a
reply value to the thread.
The reply value determines how the thread proceeds.
In this section, we first review the use operators, which are concerned
with threads making such use of services, and then extend the process
extraction operation to the use operators.
The use operators can be used in combination with the thread extraction
operation from Section~\ref{sect-thread-extr} to describe the behaviour
produced by instruction sequences that make use of services.

\subsection{Use Operators}
\label{subsect-use-operator}

A \emph{service} $H$ consists of
\begin{iteml}
\item
a set $S$ of \emph{states};
\item
an \emph{effect} function $\funct{\eff}{\Meth \x S}{S}$;
\item
a \emph{yield} function
$\funct{\yld}{\Meth \x S}{\Nat}$;
\item
an \emph{initial state} $s_0 \in S$;
\end{iteml}
satisfying the following condition:
\begin{ldispl}
\Forall{m \in \Meth, s \in S}
{(\yld(m,s) = 0 \Implies
  \Forall{m' \in \Meth}{\yld(m',\eff(m,s)) = 0})}\;.
\end{ldispl}
The set $S$ contains the states in which the service may be, and the
functions $\eff$ and $\yld$ give, for each method $m$ and state $s$, the
state and reply, respectively, that result from processing $m$ in state
$s$.
By the condition imposed on services, once the service has returned $0$
as reply, it keeps returning $0$ as reply.

Let $H = \tup{S,\eff,\yld,s_0}$ be  a service and let $m \in \Meth$.
Then
the \emph{derived service} of $H$ after processing $m$, written
$\derive{m}H$, is the service $\tup{S,\eff,\yld,\eff(m,s_0)}$; and
the \emph{reply} of $H$ after processing $m$, written $H(m)$, is
$\yld(m,s_0)$.

When a thread makes a request to the service to process $m$:
\begin{iteml}
\item
if $H(m) \neq 0$, then the request is accepted, the reply is $H(m)$, and
the service proceeds as $\derive{m}H$;
\item
if $H(m) = 0$, then the request is rejected and the service proceeds as
a service that rejects any request.
\end{iteml}

We introduce the sort $\Serv$ of \emph{services} and,
for each $f \in \Foci$, the binary \emph{use} operator
$\funct{\use{\ph}{f}{\ph}}{\Thr \x \Serv}{\Thr}$.
The axioms for these operators are given in Table~\ref{axioms-use}.%
\begin{table}[!t]
\caption{Axioms for use operators}
\label{axioms-use}
\begin{eqntbl}
\begin{saxcol}
\use{\Stop}{f}{H} = \Stop                           & & \axiom{U1} \\
\use{\DeadEnd}{f}{H} = \DeadEnd                     & & \axiom{U2} \\
\use{(\pcc{x}{\Tau}{y})}{f}{H} =
\pcc{(\use{x}{f}{H})}{\Tau}{(\use{y}{f}{H})}        & & \axiom{U3} \\
\use{(\pcc{x}{g.m}{y})}{f}{H} =
\pcc{(\use{x}{f}{H})}{g.m}{(\use{y}{f}{H})}
                                      & \mif f \neq g & \axiom{U4} \\
\use{(\pcc{x}{f.m}{y})}{f}{H} = \Tau \bapf (\use{x}{f}{\derive{m}H})
                                      & \mif H(m) = 1 & \axiom{U5} \\
\use{(\pcc{x}{f.m}{y})}{f}{H} = \Tau \bapf (\use{y}{f}{\derive{m}H})
                                      & \mif H(m) = 2 & \axiom{U6} \\
\use{(\pcc{x}{f.m}{y})}{f}{H} = \DeadEnd
                                      & \mif H(m) = 0 & \axiom{U7} \\
\multicolumn{2}{@{}l@{\;\;}}{
\use{(\pcc{x}{\ac(e_1,\ldots,e_n)}{y})}{f}{H} =
\pcc{(\use{x}{f}{H})}{\ac(e_1,\ldots,e_n)}{(\use{y}{f}{H})}}
                                                      & \axiom{U8} \\
\multicolumn{2}{@{}l@{\;\;}}{
\use{\pcs{k}{\Tau}{x_1,\ldots,x_k}}{f}{H} =
\pcs{k}{\Tau}{\use{x_1}{f}{H},\ldots,\use{x_k}{f}{H}}}
                                                      & \axiom{U9} \\
\use{\pcs{k}{g.m}{x_1,\ldots,x_k}}{f}{H} =
{} \\ \qquad
\pcs{k}{g.m}{\use{x_1}{f}{H},\ldots,\use{x_k}{f}{H}}
              & \mif f \neq g                         & \axiom{U10} \\
\use{\pcs{k}{f.m}{x_1,\ldots,x_k}}{f}{H} =
\Tau \bapf (\use{x_i}{f}{\derive{m}H})
              & \mif H(\seq{m}) = i \land i \in [1,k] & \axiom{U11} \\
\use{\pcs{k}{f.m}{x_1,\ldots,x_k}}{f}{H} = \DeadEnd
              & \mif H(\seq{m}) \notin [1,k]          & \axiom{U12} \\
\multicolumn{2}{@{}l@{\;\;}}{
\use{\pcs{k}{\ac(e_1,\ldots,e_n)}{x_1,\ldots,x_k}}{f}{H} =
\pcs{k}{\ac(e_1,\ldots,e_n)}{\use{x_1}{f}{H},\ldots,\use{x_k}{f}{H}}}
                                                      & \axiom{U13} \\
\proj{n}{\use{x}{f}{H}} = \proj{n}{\use{\proj{n}{x}}{f}{H}}
                                                    & & \axiom{U14}
\end{saxcol}
\end{eqntbl}
\end{table}
Intuitively, $\use{p}{f}{H}$ is the thread that results from processing
all actions performed by thread $p$ that are of the form $f.m$ by
service $H$.
When a basic action of the form $f.m$ performed by thread $p$ is
processed by service $H$, it is turned into the basic action $\Tau$ and
postconditional composition or postconditional switch is removed in
favour of basic action prefixing on the basis of the reply value
produced.

We add the use operators to \PGAmrpc\ as well.
We will only use the extension in combination with the thread extraction
operation $\textr{\ph}$ and define
$\textr{\use{P}{f}{H}} = \use{\textr{P}}{f}{H}$.
Hence, $\textr{\use{P}{f}{H}}$ denotes the thread produced by $P$ if $P$
makes use of $H$.
If $H$ is a service such as an unbounded counter, an unbounded stack or
a Turing tape, then a non-regular thread may be produced.

\subsection{Extending Process Extraction to the Use Operators}
\label{subsect-process-extr-use}

In order to extend the process extraction operation to the use
operators, we need an extension of \ACPt\ with action renaming
operators.
The unary action renaming operator $\aren{R}$, for
$\funct{R}{\Actt}{\Actt}$ such that $R(\tau) = \tau$, can be explained
as follows: $\aren{R}(p)$ behaves as $p$ with each atomic action
replaced according to $R$.
The axioms for action renaming are given in~\cite{Fok00}.
We write $\aren{e' \mapsto e''}$ for the renaming operator $\aren{R}$
with $R$ defined by $R(e') = e''$ and $R(e) = e$ if $e \neq e'$.

We also need to define a set $A_f \subseteq \Act$ and a function
$\funct{R_f}{\Actt}{\Actt}$ for each $f \in \Foci$:
\begin{ldispl}
A_f =
\set{\snd_f(d) \where d \in \Meth \union \Nat} \union
\set{\rcv_f(d) \where d \in \Meth \union \Nat}\;;
\end{ldispl}%
for all $e \in \Actt$, $m \in \Meth$ and $r \in \Nat$:
\begin{ldispl}
\begin{aeqns}
R_f(\snd_\serv(r)) = \snd_f(r)\;, \\
R_f(\rcv_\serv(m)) = \rcv_f(m)\;, \\
R_f(e)             = e
 & & \mif \AND{r' \in \Nat}{e \neq \snd_\serv(r')} \land
          \AND{m' \in \Meth}{e \neq \rcv_\serv(m')}\;.
\end{aeqns}
\end{ldispl}

The additional defining equations for the process extraction operation
concerning the use operators are given in
Table~\ref{eqns-process-extr-add},%
\begin{table}[!t]
\caption{Additional defining equations for process extraction operation}
\label{eqns-process-extr-add}
\begin{eqntbl}
\begin{eqncol}
\cpextr{\use{t}{f}{H}} =
\aren{\stp^* \mapsto \stp}
 (\encap{\set{\stp,\ol{\stp}}}
   (\encap{A_f}(\cpextr{t} \parc \aren{R_f}(\cpextr{H}))))
\eqnsep
\cpextr{H} =
\rec{X_H}
 {\{X_{H'} =
  \vAltc{m \in \Meth}
   \rcv_\serv(m) \seqc \snd_\serv(H'(m)) \seqc X_{\derive{m} H'}
   \altc \ol{\stp} \where H' \in \rDelta(H)\}}
\vspace*{.5ex}
\end{eqncol}
\end{eqntbl}
\end{table}
where $\rDelta(H)$ is inductively defined as follows:
\begin{iteml}
\item
$H \in \rDelta(H)$;
\item
if $m \in \Meth$ and $H' \in \rDelta(H)$, then
$\derive{m} H' \in \rDelta(H)$.
\end{iteml}

The extended process extraction operation preserves the axioms for the
use operators.
Owing to the presence of axiom schemas with semantic side conditions in
Table~\ref{axioms-use}, the axioms for the use operators include proper
axioms and axioms that have a semantic side condition of the form
$H(m) = n$.
By that, the precise formulation of the preservation result is somewhat
complicated.
\begin{proposition}
\label{prop-preservation-axioms-use}
\mbox{}
\begin{enuml}
\item
Let $\phi$ be a proper axiom for the use operators.
Then $\transl{\phi}$ is derivable from the axioms of \ACPt\ with action
renaming and guarded recursion.
\item
Let $\phi \;\mif \psi$ be an axiom with semantic side condition
for the use operators.
Then $\transl{\phi}$ is derivable from the axioms of \ACPt\ with action
renaming and guarded recursion under the assumption that $\psi$
holds.
\end{enuml}
\end{proposition}
\begin{proof}
The proof is straightforward.
We sketch the proof for axiom U4, writing $E_H$ for
$\{X_{H'} =
   \vAltc{m \in \Meth}
    \rcv_\serv(m) \seqc \snd_\serv(H'(m)) \seqc X_{\derive{m} H'}
   \altc \ol{\stp} \where H' \in \rDelta(H)\}$.
By the definition of the process extraction operation, it is sufficient
to show that
$\cpextr{\use{(\pcc{x}{f.m}{y})}{f}{H}} =
 \cpextr{\Tau \bapf (\use{x}{f}{\derive{m}H})}$
is derivable under the assumption that $H(m) = 1$ holds.
In outline, this goes as follows:
\begin{trivlist}
\item
$
\begin{geqns}
\cpextr{\use{(\pcc{x}{f.m}{y})}{f}{H}}
\\ \; {} =
\aren{\stp^* \mapsto \stp}
 (\encap{\set{\stp,\ol{\stp}}}
   (\encap{A_f}
     (\snd_f(m) \seqc
      (\rcv_f(1) \seqc x \altc \rcv_f(2) \seqc y) \parc
      \aren{R_f}(\rec{X_H}{E_H}))))
\\ \; {} =
\iact \seqc \iact \seqc
\aren{\stp^* \mapsto \stp}
 (\encap{\set{\stp,\ol{\stp}}}
   (\encap{A_f}(x \parc \aren{R_f}(\rec{X_{\derive{m}H}}{E_H}))))
\\ \; {} =
\cpextr{\Tau \bapf (\use{x}{f}{\derive{m}H})}\;.
\end{geqns}
$
\end{trivlist}
In the first and third step, we apply defining equations of
$\cpextr{\ph}$.
In the second step, we apply axioms of \ACPt\ with action renaming and
guarded recursion, and use the assumption that $H(m) = 1$.
\qed
\end{proof}

Let $P$ be a closed \PGAmrpc\ term and $H$ be a service.
Then $\pextr{\textr{\use{P}{f}{H}}}$ denotes the process produced by $P$
if $P$ makes use of $H$.
Instruction sequences that make use of services such as unbounded
counters, unbounded stacks or Turing tapes are interesting because they
may produce non-regular processes.

\section{\PGLDmr\ Programs and the Use of Boolean Registers}
\label{sect-PGLDmr-BR}

In this section, we show that all regular processes can also be produced
by programs written in a program notation which is close to existing
assembly languages, and even by programs in which no atomic action
occurs more than once.
The latter result requires programs that make use of Boolean registers.

\subsection{The Program Notation \PGLDmr}
\label{subsect-PGLDmr}

A hierarchy of program notations rooted in program algebra is introduced
in~\cite{BL02a}.
One program notation that belongs to this hierarchy is \PGLD, a very
simple program notation which is close to existing assembly languages.
It has absolute jump instructions and no explicit termination
instruction.
Here, we introduce \PGLDmr, an extension of \PGLD\ with multiple-reply
test instructions.

In \PGLDmr, like in \PGAmr, it is assumed that there is a fixed but
arbitrary finite set of \emph{basic instructions} $\BInstr$.
The primitive instructions of \PGLDmr\ differ from the primitive
instructions of \PGAmr\ as follows: for each $l \in \Nat$, there is
an \emph{absolute jump instruction}~$\ajmp{l}$ instead of a forward jump
instruction~$\fjmp{l}$.
\PGLDmr\ programs have the form $u_1;\ldots;u_k$, where $u_1,\ldots,u_k$
are primitive instructions of \PGLDmr.

The effects of all instructions in common with \PGAmr\ are as in \PGAmr\
with one difference: if there is no next instruction to be executed,
termination occurs.
The effect of an absolute jump instruction $\ajmp{l}$ is that execution
proceeds with the $l$-th instruction of the program concerned.
If $\ajmp{l}$ is itself the $l$-th instruction, then deadlock occurs.
If $l$ equals $0$ or $l$ is greater than the length of the program, then
termination occurs.

We define the meaning of \PGLDmr\ programs by means of a function
$\pgldmrpga$ from the set of all \PGLDmr\ programs to the set of all
closed \PGAmr\ terms.
This function is defined by
\begin{ldispl}
\pgldmrpga(u_1 \conc \ldots \conc u_k) =
(\phi_1(u_1) \conc \ldots \conc \phi_k(u_k) \conc
 \halt \conc \halt)\rep\;,
\end{ldispl}%
where the auxiliary functions $\phi_j$ from the set of all primitive
instructions of \PGLDmr\ to the set of all primitive instructions of
\PGAmr\ are defined as follows ($1 \leq j \leq k$):
\begin{ldispl}
\begin{aceqns}
\phi_j(\ajmp{l}) & = & \fjmp{l-j}       & \mif j \leq l \leq k\;, \\
\phi_j(\ajmp{l}) & = & \fjmp{k+2-(j-l)} & \mif 0   <  l   <  j\;, \\
\phi_j(\ajmp{l}) & = & \halt            & \mif l = 0 \lor l > k\;, \\
\phi_j(u)        & = & u
                    & \mif u\; \mathrm{is\;not\;a\;jump\;instruction}\;.
\end{aceqns}
\end{ldispl}

\PGLDmr\ is as expressive as \PGAmr.
\begin{proposition}
\label{prop-express-PGLDmr}
For each closed \PGAmr\ term $P$, there exists a \PGLDmr\ program $P'$
such that $\textr{P} =\textr{\pgldmrpga(P')}$.
\end{proposition}
\begin{proof}
In~\cite{BL02a}, a number of functions (called embeddings in that paper)
are defined, whose composition gives, for each closed \PGA\ term $P$, a
\PGLD\ program $P'$ such that $\textr{P} =\textr{\pgldpga(P')}$, where
$\pgldpga$ is the restriction of $\pgldmrpga$ to \PGLD\ programs.
The extensions of the above-mentioned embeddings to cover multiple-reply
test instructions are trivial because the embeddings change only jump
and termination instructions.
\qed
\end{proof}

Below, we will write \PGLDmrpc\ for the version of \PGLDmr\ in which the
additional assumptions relating to $\BInstr$ mentioned in
Section~\ref{subsect-pc-tsi} are made.
As a corollary of Theorem~\ref{theorem-expressiveness} and
Proposition~\ref{prop-express-PGLDmr}, we have that all regular
processes over $\AAct$ can be produced by \PGLDmrpc\ programs.
\begin{corollary}
\label{corollary-express-PGLDmr}
For each process $p$ that is regular over $\AAct$, there exists a
\PGLDmrpc\ program $P$ such that $p$ is the process denoted by
$\pextr{\textr{\pgldmrpga(P)}}$.
\end{corollary}

\subsection{\PGLDmr\ Programs Acting on Boolean Registers}
\label{subsect-PGLDmr-BR}

First, we describe services that make up Boolean registers.

A Boolean register service accepts the following methods:
\begin{itemize}
\item
a \emph{set to true method} $\setbr{\True}$;
\item
a \emph{set to false method} $\setbr{\False}$;
\item
a \emph{get method} $\getbr$.
\end{itemize}
We write $\Methbr$ for the set
$\set{\setbr{\True},\setbr{\False},\getbr}$.
It is assumed that $\Methbr \subseteq \Meth$.

The methods accepted by Boolean register services can be explained as
follows:
\begin{itemize}
\item
$\setbr{\True}$\,:
the contents of the Boolean register becomes $\True$ and the reply is
$\True$;
\item
$\setbr{\False}$\,:
the contents of the Boolean register becomes $\False$ and the reply is
$\False$;
\item
$\getbr$\,:
nothing changes and the reply is the contents of the Boolean register.
\end{itemize}

Let $s \in \set{\True,\False,\Blocked}$.
Then the \emph{Boolean register service} with initial state $s$, written
$\BR_s$, is the service $\tup{\set{\True,\False,\Blocked},\eff,\eff,s}$,
where the function $\eff$ is defined as follows
($b \in \set{\True,\False}$):
\begin{ldispl}
\begin{geqns}
\eff(\setbr{\True},b) = \True\;,\;
\\
\eff(\setbr{\False},b) = \False\;,
\\
\eff(\getbr,b) = b\;,
\end{geqns}
\qquad\qquad
\begin{geqns}
\eff(m,b) = \Blocked & \mif m \not\in \Methbr\;,
\\
\eff(m,\Blocked) = \Blocked\;.
\end{geqns}
\end{ldispl}
Notice that the effect and yield functions of a Boolean register service
are the same.

We have that, by making use of Boolean registers, \PGLDmrpc\ programs in
which no atomic action from $\AAct$ occurs more than once can produce
all regular processes over $\AAct$.
\begin{theorem}
\label{corollary-express-BR}
For each process $p$ that is regular over $\AAct$, there exists a
\PGLDmrpc\ program $P$ in which each atomic action from $\AAct$ occurs
no more than once such that $p$ is the process denoted by
$\pextr{( \ldots
         (\textr{\pgldmrpga(P)}
           \useop{\br{1}} \BR_\False) \ldots
           \useop{\br{k}} \BR_\False)}$,
where $k$ is the length of $P$.
\end{theorem}
\begin{proof}
By the proof of Theorem~\ref{theorem-completeness} given in
Section~\ref{sect-expressiveness}, it is sufficient to show that, for
each thread $p$ that is regular over $\BAct$, there exist a \PGLDmr\
program $P$ in which each atomic action from $\BAct$ occurs no more than
once and a $k \in \Natpos$ such that $p$ is the thread denoted by
$( \ldots
         (\textr{\pgldmrpga(P)}
           \useop{\br{1}} \BR_\False) \ldots
           \useop{\br{k}} \BR_\False)$.

Let $p$ be a thread that is regular over $\BAct$.
We may assume that $p$ is produced by a \PGLDmr\ program $P'$ of the
following form:
\begin{ldispl}
\ptst{a_1} \conc
 \ajmp{(3 \mul k_1 + 1)} \conc \ajmp{(3 \mul k'_1 + 1)} \conc {}
\\
\quad \vdots
\\
\ptst{a_n} \conc
 \ajmp{(3 \mul k_n + 1)} \conc \ajmp{(3 \mul k'_n + 1)} \conc {}
\\
\ajmp{0} \conc \ajmp{0} \conc \ajmp{0} \conc \ajmp{(3 \mul n + 4)}\;,
\end{ldispl}
where, for each $i \in [1,n]$, $k_i, k'_i \in [0,n - 1]$
(cf.\ the proof of Proposition~2 from~\cite{PZ06a}).
It is easy to see that the \PGLDmr\ program $P$ that we are looking for
can be obtained by transforming $P'$: by making use of $n$ Boolean
registers, $P$ can distinguish between different occurrences of the same
basic instruction in $P'$, and in that way simulate $P'$.
\qed
\end{proof}

\section{Conclusions}
\label{sect-concl}

Because process algebra is considered relevant to computer science,
there must be programmed systems whose behaviours are taken for
processes as considered in process algebra.
In that light, we have investigated the connections between programs and
the processes that they produce, starting from the perception of a
program as a single-pass instruction sequence.
We have shown that, by apposite choice of basic instructions, all
regular processes can be produced by single-pass instruction sequences
as considered in program algebra.

We have also made precise what processes are produced by threads that
make use of services.
The reason for this is that single-pass instruction sequences under
execution are regular threads and regular threads that make use of
services such as unbounded counters, unbounded stacks or Turing tapes
may produce non-regular processes.
An option for future work is to characterize the classes of processes
that can be produced by single-pass instruction sequences that make use
of such services.

\bibliographystyle{spmpsci}
\bibliography{TA}


\end{document}